\documentclass[twocolumn, pra, superscriptaddress]{revtex4}

\usepackage{hyperref,graphicx,amsmath,latexsym,revsymb,amssymb,verbatim,color}
\usepackage[mathletters]{ucs}
\usepackage[utf8x]{inputenc}
\usepackage[T1]{fontenc}
\usepackage{hyperref,graphicx,amsmath,latexsym,revsymb,amssymb,verbatim,color}
\usepackage{amsthm}
\usepackage{stmaryrd}
\usepackage[noend]{algpseudocode}
\usepackage{multirow}
\usepackage{algorithm}

\DeclareMathOperator\lcm{lcm}
\DeclareMathOperator\ord{ord}

\newtheorem{lemma}{Lemma}

\begin{document}

\title{Factoring Safe Semiprimes with a Single Quantum Query}
\author{Frédéric Grosshans} 
\affiliation{Laboratoire Aimé Cotton, CNRS, Univ. Paris-Sud, ENS Cachan,
Univ. Paris-Saclay, 91405 Orsay, France}
\author{Thomas Lawson} \affiliation{LTCI --- T\'el\'ecom ParisTech, 23 avenue d'Italie, 75013, Paris, France}
\author{François Morain} \affiliation{École Polytechnique/LIX and
Institut national de recherche en informatique et en automatique (INRIA)}
\author{Benjamin Smith} \affiliation{INRIA and École Polytechnique/LIX}

\date{\today}
\begin{abstract}
Shor's factoring algorithm (SFA), by its ability to efficiently factor large numbers, has the potential to undermine contemporary encryption. 
At its heart is a process called order finding, which quantum mechanics lets us perform efficiently. 
SFA thus consists of a \emph{quantum order finding algorithm} (QOFA), bookended by classical routines which, given the order, return the factors. 
But, with probability up to $1/2$, these classical routines fail, and QOFA must be rerun. 
We modify these routines using elementary results in number theory, improving the likelihood that they return the factors.

The resulting quantum factoring algorithm is better than SFA at factoring safe semiprimes, an important class of numbers used in cryptography. 
With just one call to QOFA, our algorithm almost always factors safe semiprimes. 
As well as a speed-up, improving efficiency gives our algorithm other, practical advantages: unlike SFA, it does not need a randomly picked input, making it simpler to construct in the lab; and in the (unlikely) case of failure, the same circuit can be rerun, without modification.

We consider generalizing this result to other cases, although we do not find a simple extension, and conclude that SFA is still the best algorithm for general numbers (non safe semiprimes, in other words). Even so, we present some simple number theoretic tricks for improving SFA in this case. 

\end{abstract} 

\maketitle

\section{Introduction} 
Shor's factoring algorithm (SFA) \cite{Shor1994} 
is a promising application of quantum computers. 
Like classical 
factoring algorithms, 
it returns a nontrivial factor of an integer $N= \prod_{i=1}^{k} p_i ^{e_i}$, 
where the $p_i$ are primes, with $p_i \neq p_j$ and $e_i > 0$. 
But, unlike its classical counterparts, 
SFA does this efficiently---its runtime increases just polynomially 
in the length of $N$, 
as opposed to subexponentially for the best classical algorithms. 
This speed-up is due to the quantum subroutine at the heart of SFA, 
the quantum order finding algorithm (QOFA).  
QOFA efficiently calculates 
the \emph{order} of an integer $a$ modulo
$N$: that is, given coprime integers~$a$ and~$N$,
QOFA returns 
the smallest integer $r > 0$ such that $a^r \bmod{N} = 1$, 
where $\bmod$ is the \emph{modulo operator}. 
(In practice, QOFA is only guaranteed to return a divisor of the order.)

SFA proceeds as follows:
an integer $a$ is chosen at random.
In the unlikely event that $\gcd(a,N)$ 
(the \emph{greatest common divisor} of $a$ and $N$)
is greater than~$1$, then we have already found a nontrivial factor of $N$.
Otherwise, we use QOFA to find the order $r$ of $a$ modulo $N$.
If $r$ is even (which occurs with probability at least $1/2$), then
a factor of $N$ may be given by $\gcd (a^{r/2} \pm 1, N)$.
If a factor is not found, or if $r$ is odd,
then the procedure is relaunched with a different random value of $a$ \cite{Shor1994, Shor1997, nielsen+chuang2000}.

Much work has been done on the quantum part of the algorithm \cite{Vandersypen2001, lu-prl-99-250504, Lanyon2007, Mathews2009, eml-nphot-6-773, lu-nphys-10-719, Monz16}. 
In this letter we take a different approach: 
we reduce the probability of failure by applying elementary results in number theory to the \emph{classical} part of SFA. 

Our main result concerns \emph{safe semiprimes},
defined here to be the product $N=p_1 p_2$ of two safe primes $p_1$ and $p_2$ of the form  $p_i = 2q_i +1$, where $q_i$ is prime.
Safe primes were introduced by Sophie Germain in her study of Fermat's
last theorem \cite{Laubenbacher10}.

They are now  called \emph{safe} because of their relation to 
\emph{strong primes}---primes $p$ such that $p-1$ has at least 
a large prime factor.  
The \emph{strong} and \emph{safe} name of these primes
comes from the hardness to factor their products
using Pollard's ``$p-1$'' factoring algorithm \cite{Pollard1974}.
The safe primes are the strongest of all strong prime, and are 
well defined in terms of number theory.
Until the discovery of elliptic curve factorization \cite{Lenstra87}, 
they were widely held as being the hardest to factor for this reason.
This led to the recommendation of using the related strong primes for RSA
cryptography, that is prime such that $p-1$ has large prime factors
but this recommendation is contested \cite{RivestSilverman00}, 
even if they are still recommended in the 
National Institute of Standard’s FIPS 186-4 \cite{FIPS186-4}.

Independently of their relationship to strong semiprimes, safe semiprimes are 
used in the Cramer--Shoup signature scheme \cite{CramerShoup00}, where the 
primality  of $q_1$ and $q_2$ allows to reduce the cost of the signature 
generation and, more importantly, to prove the security of the scheme 
against an adaptive chosen message attack under the
strong RSA assumption. 
Both properties come from the fact that a random hash $h$ generates the group
of quadratic residues modulo $N$ with overwhelming probability.

We present an algorithm which, given one call to QOFA, factors safe semiprimes \(N\) with probability $ \approx 1 - 4 / N$. 
In contrast, for a single query of QOFA, SFA's classical routines fail to factor these numbers half the time. 

At the core of our algorithm is the simple observation that for these
numbers, the output \(d\) of QOFA always 
divides \(2q_1q_2\), and all divisors of \(2q_1q_2\), except $1$ and $2$,
allow to compute $p_1$ and $p_2$.
This information, contained in $d$, is not used optimally by SFA.
Interestingly, the property that makes safe semiprimes the hardest
numbers to factor using Pollard's classical $p-1$ algorithm
is exactly that which makes them easy to factor with our quantum algorithm. 
We also consider moving from the special case of safe semiprimes to general numbers. Although the proofs do not extend naturally, we find a few tricks that may make factoring general numbers easier.

The idea of modifying the classical parts of SFA so as to reduce the
dependence on QOFA is not new, although previous studies 
fine-tune the algorithm, rather than re-think part of it, as we try to do.
References \cite{Leander2002, Markov2012, Lawson2014} discuss 
how best to pick $a$ to maximize the probability of factoring. 
References \cite{Shor1997, Knill1995} increase efficiency 
by extracting~$r$ from failed outputs of QOFA. 
These results offer practical advantages beyond simple 
theoretical improvements:
although efficient in theory, in practice each call to QOFA needs an individual experiment, which is time-consuming and costly in resources.
Indeed, for some architectures---photonics, for instance---each experiment requires a new circuit.

Beyond its low failure rate, our algorithm has two advantages when factoring safe semiprimes (and possibly
a larger class of numbers). The first is that
rather than picking the value of $a$ randomly, we can keep it constant. 
In this letter we use $a=2$, 
but other sensible choices can be made.
The second is that, when our algorithm does fail, exactly the same
circuit can be used again, unlike SFA which needs to be reset 
with a different value of \(a\). 

We note that this is not the first time specific values of $a$ have been used in investigations of SFA. 
Experimentalists, held back by young technology, have so far picked values of $a$ that let them simplify the circuit and focus on certain interesting parts 
of the algorithm \cite{Vandersypen2001, lu-prl-99-250504, Lanyon2007, Mathews2009, lu-nphys-10-719, eml-nphot-6-773, Monz16}. 
Reference \cite{Smolin2013} highlighted the danger of this approach, arguing that the value of $a$ should not rely on \emph{knowledge} of the factors. 
We emphasize that the success of our algorithm does not rely on knowing the factors in advance:
our value \(a=2\) is obviously independent of the number being
factored.


\section{Number theoretic tools}

We write \(x\mid y\) to mean \(x\) is a \emph{divisor} of \(y\):
that is, there exists an integer \(z\) such that \(y = xz\).
A positive integer is \emph{prime} if it has no 
divisors other than itself and~\(1\).
If \(y_1,\ldots,y_k\) are integers,
then \(\gcd(y_1,\ldots,y_k)\) denotes their greatest common divisor:
that is, the largest positive integer \(x\)
such that \(x\mid y_i\) for all \(1 \le i \le k\).
Dually, \(\lcm(y_1,\ldots,y_k)\) denotes the least common multiple 
of the \(y_i\):
that is, the smallest positive integer \(x\)
such that \(y_i\mid x\) for all \(1 \le i \le k\).

Two concepts from elementary number theory 
will be essential throughout this letter. 
In the following
\[
    N = \prod_{i=1}^k p_i^{e_i}
    \ ,
\]
where the \(p_i\) are distinct primes.
The first important concept is the Euler $\phi$ function, 
\begin{align}
    \label{eq:phi-def}
    \phi(N) = \prod_{i=1}^k p_i^{e_i -1} (p_i-1), 
\end{align}
which gives the number of integers \(x\) in \([1,N]\) such that
$\gcd(x,N) = 1$.  
The second is the Carmichael $\lambda$ function, $\lambda(N)$,  defined as the smallest integer such that
$x^{\lambda(N)} \bmod{N} = 1$
for \emph{every} $x$ satisfying $\gcd(x,N) = 1$.
It follows that if \(r\) is the order of some \(a\) modulo \(N\),
then $r\mid \lambda(N)$.
Carmichael's theorem shows that if $N$ factors as above,
then
\begin{align}
    \label{eq:lambda-def}
    \lambda(N) =  \lcm(\lambda(p_1^{e_1}),\ldots,\lambda(p_k^{e_k}))
    \ , 
\end{align}
where 
\begin{align}
    \label{eq:lambda-odd}
    \lambda(p_i^{e_i})
    &= 
    \phi(p_i^{e_i}) 
    \text{ for odd primes } p_i,\\
    \label{eq:lambda-even}
    \lambda(2^{e_i})
     &= \frac{1}{2}\phi(2^{e_i}) = 2^{e_i-2} \text{ for }  e_i > 2, 
\end{align}
and \(\lambda(2^{1})=1\), \(\lambda(2^{2})=2\).

\section{Factoring safe semiprimes}

A \emph{semiprime} is a product of two primes, hence 
finding any nontrivial factor of a semiprime
amounts to finding its complete prime factorization.
An odd prime \(p\) is called \emph{safe}
if \((p-1)/2\) is also prime
(note that \(p-1\) is always divisible by \(2\) when \(p\) is odd,
so \(p\) being safe can be understood to mean that \(p-1\) is 
``as close to prime as possible'').
A \emph{safe semiprime} is a product of two safe primes.

\subsection{Semiprimes and Euler's \(\phi\) function}
\label{sec:semiprimes}

If $N = p_1p_2$ is a semiprime \cite{RSA1978} 
(including---but not restricted to---the case of safe semiprimes),
then
knowing $\phi(N)$ is sufficient for finding a factor of $N$. 
Indeed, the factors can be found immediately and deterministically. 
In this case $\phi(N) = N +1 - (p_1+p_2)$ which, 
when solved simultaneously with $N=p_1 p_2$, gives a factor. 
Explicitly, expanding the product \((X-p_1)(X-p_2)\)
shows that $p_1$ and $p_2$ are the solutions of 
the quadratic equation
\begin{equation}
    \label{eq:quadratic-eqn}
    X^2 - (N - \phi(N) +1)X + N = 0
    \ .
\end{equation}

For general $N= \prod_{i=1}^{k} p_i ^{e_i}$, 
Miller~\cite{Miller1976} has shown that
knowing both $N$ and
$\phi(N)$ lets us efficiently factor $N$
probabilistically,
or deterministically under the extended Riemann hypothesis.

\subsection{Shor's factoring algorithm for safe semiprimes}

Let \(N = p_1p_2\) be a safe semiprime,
so \(p_1 = 2q_1 + 1\) and \(p_2 = 2q_2 + 1\)
with \(q_1\) and \(q_2\) both prime.
SFA chooses a random residue \(a\) modulo~\(N\)
and attempts to compute its order \(r\) using a call to QOFA.
SFA iterates this process until \(r\) is even and \(a^{r/2}\) is not
congruent to \(-1\) modulo \(N\);
if \(r = \ord_N(a)\), then \(\gcd(a^{r/2}-1,N)\) yields a factor of \(N\)
(either \(p_1\) or \(p_2\)).

We have \(\phi(N) = 4q_1q_2\)
and \(\lambda(N) = 2q_1q_2\).
As noted above,
if \(r\) is the order of some \(a\) modulo \(N\),
then $r\mid \lambda(N)$;
so \(r\) must be one of the eight divisors of \(2q_1q_2\).
Their frequencies are listed in
Table~\ref{tab:RSA-orders}.

\begin{table}
    \caption{The possible multiplicative orders 
        for integers modulo safe semiprimes
        \(N = p_1p_2\) 
        (where \(p_i = 2q_i + 1\) with \(p_1\), \(p_2\), \(q_1\), and
        \(q_2\) all prime). 
        The number of integers is deduced from the isomorphism between the 
        multiplicative group \(\left(\mathbb{Z}/n\mathbb{Z}\right)^{\times}\)and the product of cyclic
        groups \(C_{2q_1}\times C_{2q_2}\).
    }
    \label{tab:RSA-orders}
    \begin{center}
        \begin{ruledtabular}
        \begingroup \squeezetable 
        \begin{tabular}{ccc}
            \multirow{2}{*}{\(r\)}       & Number of integers 
                                         & \multirow{2}{*}{Restrictions on $a$}
            \\
            & \(1 \le a< N\) of order \(r\) & \\
            \hline
            \(1\)       & 1 &\(a = 1 \) \\
            \hline
            \multirow{2}{*}{\(2\)}       & \multirow{2}{*}{3} & \(a =
            N-1\) and two
            \\
             & & other $a\ge\sqrt{N+1}>2$
            \\
            \hline
            \(q_1\)     & \(q_1-1\) &$a\neq2$\\
            \(q_2\)     & \(q_2-1\) &$a\neq2$ \\
            \hline
            \(2q_1\)    & \(3(q_1-1)\)&$a\neq2$\\ 
            \(2q_2\)    & \(3(q_2-1)\)  &$a\neq2$ \\
            \hline
            \(q_1q_2\)  & \((q_1-1)(q_2-1)\) \\
            \(2q_1q_2\) & \(3(q_1-1)(q_2-1)\) \\
        \end{tabular}
        \endgroup
        \end{ruledtabular}
    \end{center}
\end{table}

Despite its name, QOFA does \emph{not} always return the order of \(a\) modulo \(N\).
In practice, the quantum Fourier transform (QFT) in QOFA 
yields an approximation of the fraction \(t/\ord_N(a)\) 
for some integer \(t\) uniformly  randomly chosen from
\(\{1,\ldots,\ord_N(a)\}\) \cite{nielsen+chuang2000}.
The quality of this approximation  
depends on the implementation of the QFT: 
the output is always a bit granular and blurred,
depending on the number of qubits \cite{Shor1997, Knill1995, nielsen+chuang2000}; 
circuits approximating the QFT are available to save resources at the expense of 
precision \cite{Coppersmith1994, CleveWatrous2000, Cheung2004}.
In all theses cases; the exact value of the output of an ideal QFT can be found
with high probability by trying values neighboring close to the output 
\cite{Shor1997,Knill1995}.
As our primary theoretical concern is independent of this 
implementation-dependent consideration, we simplify our 
theoretical analysis by assuming the implementation to 
actually give us  \(t/\ord_N(a)\) exactly.
QOFA then uses a continued fraction expansion to get 
\(d = \ord_N(a)/\gcd(t,\ord_N(a))\). 
The probability that this is the true order 
depends visibly on the number of nontrivial divisors of \(\ord_N(a)\). 

We can use Table~\ref{tab:RSA-orders} to prove that
if \(N\) is a safe semiprime and we sample \(a\)
uniformly at random, 
then each iteration of SFA
yields a factor of \(N\) with probability \(\le 1/2\).
The probability would be exactly \(1/2\)
if QOFA were a perfect order finding algorithm:
looking at Table~\ref{tab:RSA-orders}, we see that
for a given \(a\) we have \(P(r \text{ is even}) = 3/4\), 
but then \(P(a^{r/2} \not\equiv -1 \pmod{N}) = 2/3\);
so the probability that a uniformly randomly chosen \(a\)
yields a factor is \(3/4\cdot2/3 = 1/2\).
Taking QOFA imperfection into account (that is, the fact that QOFA
sometimes returns a proper divisor of \(\ord_N(a)\))
decreases this probability.
It follows that if \(N\) is a safe semiprime, 
then the expected number of iterations of SFA---and hence 
the expected number of calls to QOFA---is \(\ge 2\).
When $a=2$ is fixed and the safe semiprime is chosen randomly, SFA is 
conjectured to succeed with
probability $\tfrac12$. 
A numerical investigation of all the 960,392,564 safe primes $p_i\le 10^{12}$ 
indeed gives success probability $\tfrac12 - 2.06\times 10^{-10}$.

\subsection{Factoring safe semiprimes with one call to QOFA}

We propose a change in the classical part of SFA. This leads to
a new, simple quantum factoring algorithm
which almost always succeeds in factoring a safe semiprime 
\(N = p_1p_2 = (2q_1+1)(2q_2+1)\)
with 
exactly \emph{one} call to QOFA.

The algorithm is given in pseudocode in Algorithm~\ref{alg:alg-1}. 
It is based on  
Lemma~\ref{lemma:orders-of-2} which shows that
the only two possible values for \(\ord_N(2)\) 
are \(q_1q_2\) and \(2q_1q_2\),
and that QOFA returns a useful factor of \(\ord_N(2)\) 
with extremely high probability.

\begin{lemma}
    \label{lemma:orders-of-2}
    Let \(N = p_1p_2 = (2q_1+1)(2q_2+1)\) be a safe semiprime,
    with \(q_1\not=q_2\) and \(q_1,q_2 > 2\) (so in particular,
    \(p_1\) and \(p_2\) are distinct safe primes greater than~\(3\)).
    Then 
    \begin{enumerate}
        \item
            \(\ord_N(2)\) is either \(q_1q_2\) or \(2q_1q_2\).
        \item
            QOFA 
            yields \(d \in \{q_1,q_2,q_1q_2,2q_1,2q_2,2q_1q_2\}\)
            with probability \(1 - \frac{1}{q_1q_2} \sim 1 - \frac{4}{N}\);
            otherwise, \(d\in\{1,2\}\).
    \end{enumerate}
\end{lemma}
\begin{proof}
    Referring to Table~\ref{tab:RSA-orders},
    we know that
    there are only eight possibilities 
    for the order of an integer modulo~\(N\):
    namely, the four odd values
    \(1\), \(q_1\), \(q_2\), and \(q_1q_2\),
    and the four even values
    \(2\), \(2q_1\), \(2q_2\), and \(2q_1q_2\).
    But \(\ord_N(2)\)
    can never be \(1\) or \(2\), since \(N > 3\).
    Now, note that \(\ord_N(2) = \lcm(\ord_{p_1}(2),\ord_{p_2}(2))\)
    (since \(p_1\not=p_2\))
    and \(\ord_{p_i}(2) \mid \lambda(p_i) = 2q_i\).
    Hence, \(\ord_N(2)\) cannot be \(q_1\), \(q_2\), \(2q_1\), or \(2q_2\),
    since otherwise \(\ord_{p_2}(2)\) would in \(\{1,2\}\),
    which contradicts \(p_2 > 5\).
    The only remaining possibilities for
    \(\ord_N(2)\) are \(q_1q_2\) and \(2q_1q_2\),
    which proves the first statement.

    For the second statement,
    recall that QOFA returns \(\ord_N(2)/\gcd(t,\ord_N(2))\)
    for some uniformly randomly distributed integer \(t\)
    in \(\{1,\ldots,\ord_N(2)\}\).
    But we have just seen that the value of \(\ord_N(2)\) 
    can only be \(q_1q_2\) or \(2q_1q_2\).
    If \(\ord_N(2) = q_1q_2\) then
    QOFA only returns 1 if \(t = q_1q_2\),
    which occurs with probability \(1/(q_1q_2)\);
    it cannot return~2.
    If \(\ord_N(2) = 2q_1q_2\)
    then QOFA returns 1 if \(t = 2q_1q_2\)
    and 2 if \(t = q_1q_2\);
    the probability of \(t\) taking one of these two values is
    \(2/(2q_1q_2) = 1/(q_1q_2)\).
    Hence, regardless of the true value of \(\ord_N(2)\),
    the probability of having a ``bad'' \(t\)
    leading to a result of 1 or 2
    is \(1/(q_1q_2)\), which tends to \(4/N\) for large \(N\).
\end{proof}

Let \(d\) be the result of QOFA applied to \(2\) modulo \(N\).
Our algorithm then attempts to recover the factors of \(N\)
from \(d\). 
The cases of odd and even \(d\) are virtually identical,
and we can unify them by setting 
\(s = d\) if \(d\) is even,
and \(s = 2d\) if \(d\) is odd;
then \(s = 2\), \(2q_1\), \(2q_2\), or \(2q_1q_2\).
The case \(s = 2\) is trivial to recognise,
and we can determine whether \(s = 2q_1q_2\) or one of the \(2q_i\)
using the elementary inequality \(2q_1,2q_2 < N/3 < 2q_1q_2\).
If \(s\) is one of the \(2q_i\)
then \(s + 1\) is one of the \(p_i\)
(and \(N/(s+1)\) is the other).
If \(s\) is \(2q_1q_2\)
then we recover \(p_1\) and \(p_2\)
by using \(\phi(N) = 2s\) in Eq.~\ref{eq:quadratic-eqn}
and applying the quadratic formula:
the factors of \(N\) are \(t \pm \sqrt{t^2 - N}\),
where \(t = (N+1)/2-s\).
In the unlikely event that \(s = 2\) (probability \(\sim 4/N\),
according to Lemma~\ref{lemma:orders-of-2}, which is 
negligible even for moderate \(N\)) then we fail to factor~\(N\).

Compared to SFA, the advantages of our algorithm are:
\begin{enumerate}
    \item
        It requires only one call to \textsc{QOFA}.
    \item
        It only ever computes the order of \(2\) modulo \(N\),
        rather than randomly (or specially) chosen \(1 < a < N\);
        this may be helpful when implementing \textsc{QOFA}, 
        since multiplication by \(2\) is a simple bitshift.
    \item
        The probability of failure decreases exponentially 
        with the length (in bits) of ~\(N\).
    \item
        In the unlikely event of failure,
        then we can attempt to factor \(N\) by simply running 
        \emph{exactly the same algorithm} again.
        In contrast to SFA, rather than changing the base, 
        we can repeat with the convenient \(a = 2\). 
\end{enumerate}
In particular, our algorithm 
is never slower than SFA.
Indeed, the only way that SFA can beat it 
is to never call QOFA,
which means the very first randomly chosen~\(a\) 
must be a factor of~\(N\);
but this is extremely improbable for even moderate-sized~\(N\).

Our algorithm and its analysis imply that safe semiprimes,
(which are considered the hardest inputs for classical factoring
algorithms)
are among the easiest inputs for quantum factoring algorithms.

\begin{algorithm}
    \caption{
        Given a safe semiprime \(N = p_1p_2\), where
        \(p_1 = 2q_1 + 1\) and \(p_2 = 2q_2 + 1\) 
        with \(p_1\), \(p_2\), \(q_1\), and \(q_2\) all distinct primes,
        returns \(\{p_1,p_2\}\) or \texttt{Failure} 
    }
    \label{alg:alg-1}
    \begin{algorithmic}[1] 
        \Function{FactorSafeSemiprime}{$N$}
        \If{\(5\mid N\)}
            \Return{\(\{5,N/5\}\)}
        \EndIf
        \State \(d \gets \textsc{Qofa}(2,N)\)
        \If{\(2\mid d\)}
	    \(s \gets d\)
        \Else{ }
         \(s \gets 2d\)
        \EndIf
        \If{\(s = 2\)}
	    \Return{\texttt{Failure}} \;
            \Comment{Prob. \(1/(q_1q_2)\)}
        \EndIf
        \If{\(s < N/3\)} \(p \gets s + 1\)
            \Comment{\(s = 2q_1\) or \(2q_2\)}
        \Else
            \Comment{\(s = 2q_1q_2\)}
            \State \(t \gets (N+1)/2 - s\)
            \State \(p \gets t + \sqrt{t^2 - N}\)
        \EndIf
        \State \Return{\(\{p,N/p\}\)}
        \EndFunction
    \end{algorithmic}
\end{algorithm}

\section{Remarks on factoring other integers}

\subsection{The restriction of Algorithm~\ref{alg:alg-1}
to safe semiprimes}

Algorithm~\ref{alg:alg-1} only works for safe semiprimes:
indeed, its correctness depends on 
the extremely small number of divisors of \(\lambda(N)\) 
when \(N\) is a safe semiprime.
It does not generalize to more general odd composites \(N\),
nor even to products of two primes where only one prime is safe,
without further information on the factors of~\(N\).

For example, suppose \(N = p_1p_2\)
where \(p_1 = 2q_1 + 1\) with \(q_1\) prime 
and \(p_2 = 2q_2q_2' + 1\) where \(q_2\) and \(q_2'\) are prime
(that is, \(p_1\) is safe and \(p_2\) is ``almost safe'').
Then \(\lambda(N) = 2q_1q_2q_2'\), which has sixteen divisors,
so there are sixteen possible orders for integers modulo~\(N\).
Suppose QOFA returns the correct order \(r\) of \(2\) 
(or of any other \(a\)) modulo \(N\).
The most common element order is \(2q_1q_2q_2'\);
but this value of \(r\) does \emph{not} immediately yield factors of \(N\).
If \(r = q_1\),
then \((2r + 1)\) is a factor of~\(N\);
but if \(r = q_2\) or \(q_2'\),
then we generally cannot derive a factor from \(r\).
The problem is that we cannot distinguish between
the cases \(r = q_1\), \(q_2\), and \(q_2'\)
without further information (such as the relative sizes
of the \(p_i\) or \(q_i\)).  

The situation degenerates further as the number of 
prime factors of \(\lambda(N)\) increases,
and indeed as the number of prime factors \(p_i\) of \(N\) increases.
Similarly, the accuracy of QOFA as an order-finding routine
is reduced
(increasing the expected runtime of SFA)
when \(\lambda(N)\) has a more complicated factorization.

\subsection{Simple improvements to SFA for general integers}

For general integers \(N\), therefore,
we return to SFA.
We note, however, that the expected number of calls to QOFA in SFA
for general integers \(N\) can be reduced with two simple
modifications.
First, once the putative order~\(r\) of \(a\) modulo \(N\) is
calculated, we can easily compute \(\gcd(r,N)\), which, 
if not \(1\),  is a nontrivial factor of~\(N\).  
We note that this happens very often for \(N\) with a square factor
(though it never happens for any RSA modulus).

Second, rather than restricting our search to even \(r\),
we could check whether \(r\) is a multiple of a small prime~\(\ell\)
(for example, \(\ell\) in \(\{2,3,5\}\));
if \(a^{r/\ell}\) is not \(1\) modulo all of the factors of \(N\),
then it leads to a nontrivial factor of \(N\).
This essentially amounts to trial division of \(r\).

These tricks, while unlikely to reveal the factors, are far less computationally expensive than rerunning QOFA. They may be worth trying when an iteration of SFA has not found the factors, before recalling QOFA.
Simulating SFA
for composite numbers \(N\) in $[10, 10^8]$
shows that these cheap improvements reduce the proportion
of such \(N\)  
that escape factoring 
after a single correct call to QOFA with \(a = 2\) 
by a factor of 4 (from 6\% to 1.5\%).
For 60\% of the remaining failures, 
we obtain $\ord_N(2)=2^{-k} \lambda(N)$, from which we can
find $\lambda(N)$ and then 
factor $N$ using Miller's algorithm \cite{Miller1976}. This gives a global efficiency gain over SFA of an order of magnitude. 
%

\section{Conclusion}

Changing the classical part of SFA without changing its quantum part, 
we have proposed
a new algorithm factoring safe semiprimes, 
 numbers used in cryptography. 
With high probability, our algorithm finds the factors from just one call to QOFA, and always outperforms SFA. 
This offers an improvement in terms of complexity on SFA, which fails to find the factors from each call to QOFA with probability 50\%,
but it also brings practical benefits. 
First, since each instance of QOFA requires a complicated physical experiment, the potential savings in runtime are large---much larger than that suggested by the improvement in computational complexity. Some architectures have particularly large overheads, as each call to QOFA requires the construction of a new circuit. 

The second benefit is that we can fix the value of $a$, rather than
picking it at random as in SFA. 
We chose $a=2$ because it slightly increases the success probability,
but this improvement is marginal and other sensible
choices exist; experimentalists may even wish to tailor a value to the circuit they are building. 

Finally, in the unlikely event that our algorithm fails, the factors can usually be found by rerunning the exact same quantum circuit, saving time in the lab. 
 
Interestingly, safe semiprimes---which were used in classical encryption precisely because they were believed to be hard to factor---are the easiest numbers to factor with a quantum computer. Two forces push in the same direction: the order is very likely to give $\lambda(N)$ (from which the factors are easily found); and QOFA performs well when given a safe semiprime, returning the order more often than when acting on a general number. 

As we move away from this class to weaker semiprimes our algorithm becomes inefficient, and we revert to SFA, 
boosted by two improvements. 
We note that even classical techniques that are unlikely to work can save time if they avoid rerunning QOFA. 

Much work has been done on improving the quantum part of SFA. We have shown that improving the classical routines, using simple classical mathematics, can speed up quantum factoring by reducing the dependence on the quantum circuit. 
Our intuition is that this is a promising path, which will yield more
improvements in efficiency.

 \bibliographystyle{ieeetr}
\bibliography{Bib}

\end{document}